\documentclass[12pt]{article}
\usepackage[left=2.5cm,right=2.5cm,top=2.5cm,bottom=2.5cm]{geometry}
\usepackage{url}            % simple URL typesetting
\usepackage{graphicx}

\usepackage{amsmath,amssymb, latexsym, amsthm, amsfonts} 

\DeclareMathOperator{\Tr}{Tr}

\newcommand{\R}{{\mathbb R}}

\newtheorem{theorem}{Theorem}
\newtheorem{definition}[theorem]{Definition}
\newtheorem{lemma}[theorem]{Lemma}
\newtheorem{proposition}[theorem]{Proposition}
\newtheorem{corollary}[theorem]{Corollary}
\newtheorem{remk}[theorem]{Remark}
\newtheorem{claim}[theorem]{Claim}
\newcommand{\myalg}[2]{
\medskip
\small{
\fbox{
\parbox{5.2in}{
\textsc{#1}: { #2}
}}
\medskip
}}
\newcommand{\E}{\textbf{E}}

\renewcommand{\Pr}[1]{\ensuremath{\mathbf{Pr}[#1]}}

\newcommand{\Ex}[1]{\ensuremath{\mathbf{E}[#1]}}

\newcommand{\freq}{\ensuremath{\text{freq}}}
\newcommand{\disc}{\ensuremath{\text{disc}}}

\newcommand{\eps}{\epsilon}

\title{Approximating the Spectrum of a Graph}
\author{David Cohen-Steiner, Weihao Kong, Christian Sohler, and Gregory Valiant}

\begin{document}

\thispagestyle{empty}

\title{Approximating the Spectrum of a Graph}
\author{
David Cohen-Steiner\thanks{INRIA Sophia-Antipolis. Email:david.cohen-steiner@inria.fr}
\and
Weihao Kong
        \thanks{Department of Computer Science, Stanford University. Email: whkong@stanford.edu}
\and
Christian Sohler
				\thanks{Department of Computer Science, TU Dortmund. Email: christian.sohler@tu-dortmund.de. The author acknowledges the support of ERC grant 307696 and of the German Science Foundation, collaborative research center 876, project A6.}
\and 
Gregory Valiant
				\thanks{Department of Computer Science, Stanford University, Email: gvaliant@cs.stanford.edu}
}
\date{}

\maketitle

\begin{abstract}
The spectrum of a network or graph $G=(V,E)$ with adjacency matrix $A$, consists of the eigenvalues of the normalized
Laplacian $L= I - D^{-1/2} A D^{-1/2}$.    This set of eigenvalues encapsulates many aspects of the structure of the graph, including the extent to which the graph posses community structures at multiple scales.  We study the problem of approximating the spectrum 
$\lambda = (\lambda_1,\dots,\lambda_{|V|})$, $0 \le \lambda_1,\le \dots, \le \lambda_{|V|}\le 2$ of $G$ in the regime where the graph is too large to explicitly calculate the spectrum.   We present a sublinear time algorithm that, given the ability to query a random node in the graph and select a random neighbor of a given node, computes a succinct representation of an approximation $\widetilde \lambda 
= (\widetilde \lambda_1,\dots,\widetilde \lambda_{|V|})$, $0 \le \widetilde \lambda_1,\le \dots, \le \widetilde 
\lambda_{|V|}\le 2$ such that $\|\widetilde \lambda - \lambda\|_1 \le \epsilon |V|$.
Our algorithm has query complexity and running time $exp(O(1/\eps))$, independent of the size of the graph, $|V|$.   We demonstrate the practical viability of our algorithm on 15 different real-world graphs from the Stanford Large Network  Dataset Collection, including social networks, academic collaboration graphs, and road networks.  For the smallest of these graphs, we are able to validate the accuracy of our algorithm by explicitly calculating the true spectrum; for the larger graphs, such a calculation is computationally prohibitive.
 
In addition we study the implications of our algorithm to property testing in the bounded degree graph model. 
We prove that if the input graphs are restricted to graphs of girth $\omega(1)$ then every $\delta$-robust 
spectral graph property is constant time testable, where a graph property is spectral if the set of
graphs in the property can be specified by their spectra, and is $\delta$-robust if the set of 
spectra consists of a core and all spectra with $l_1$-distance at most $\delta n$ to this core.
\end{abstract}

\section{Introduction}

Given an undirected graph $G=(V,E)$, its normalized Laplacian matrix
is defined as $L=I - D^{-1/2} A D^{-1/2}$, where $D$ is the diagonal matrix with entries $D_{i,i}$ given by the degree of the $i$th vertex, and $A$ is the adjacency matrix of the graph.  It is not hard to see that $L$ is positive semidefinite and singular, with eigenvalues $0=\lambda_1 \le \lambda_2 \le \ldots \le \lambda_{|V|}$, whose sum is $trace(L)$.  Many structural and combinatorial properties of graphs are exposed by the eigenvalues (and eigenvectors) of the associated graph Laplacian, $L$.  For example, as was quantified in a recent series of works~\cite{KLLOT13,LGT14,LRT12}, the value of the $i$th eigenvalue provides insights into the extent to which the graph admits a partitioning into $i$ components.   Hence the spectrum provides a detailed sense of the community structures present in the graph at multiple scales.

Inspecting the spectrum of a graph also serves as a approach to evaluating the plausibility of natural generative models for families of graphs (see, e.g.~\cite{farkas2001spectra}): for example, if the spectrum of random power-law graphs does not closely resemble the spectrum of the Twitter graph, it suggests that a random power-law graph might be a poor model for the Twitter graph.  

Given the structural information contained in the spectrum of a graph's Laplacian, it seems natural to ask the following question: \emph{How much information must one collect about a graph in order to accurately approximate its spectrum?}

\subsection{Our results}

We give the first sublinear time approximation algorithm for computing the spectrum of a 
graph $G=(V,E)$. Our algorithm assumes that we can sample vertices uniformly at random from $V$
 and that we can also query for a random neighbor of a vertex $v \in V$.  This model corresponds to assuming that we can perform a random walk in $G$, as well as randomly restart such a walk.
Our algorithm performs $exp(O(1/\eps))$ such queries to the graph and outputs an approximation $\widetilde \lambda$ of the spectrum $\lambda$ of the normalized Laplacian of $G$ (see Definition~\ref{def:normL} for the formal definition of the normalized Laplacian).

\begin{theorem}
\label{theorem:mainL}
Given the ability to select a uniformly random node from a graph $G=(V,E)$, and, for a given node, query a uniformly random neighbor of that node, then with probability at least $2/3$ one can approximate the spectrum of the normalized Laplacian of $G$ to additive error $\eps$ in earth mover distance, with runtime and number of queries bounded by $exp(O(1/\eps)).$
\end{theorem}

In the above theorem, our algorithm outputs a succinct representation of the spectrum, regarded as a discrete distributions over $[0,2]$.  This representation corresponds to approximations of of each of the $1/\eps$ quantiles of the spectrum---i.e. an approximation of the $\eps |V|$th smallest eigenvalue, the $2\eps |V|$th smallest, the $3 \eps |V|$th smallest, etc.  If desired, such a succinct representation can then be converted in linear time into a length $|V|$ vector that has $\ell_1$ distance at most $\eps|V|$ from the true vector of sorted eigenvalues of $G$.    We also note that the probability of success, $2/3$, was chosen because this is standard in the property testing literature; this probability can be trivially be boosted to any constant $\le 1$ without changing the asymptotic runtime.

Our algorithm for approximating the spectrum is based on approximating the first $O(1/\epsilon)$ ``spectral moments'', the quantities $
\frac{1}{|V|}\sum_{i=1}^{|V|} \lambda_i^\ell$ for integers $\ell = 1,2,\ldots.$ These moments are traces of matrix powers 
of the random walk matrix of $G$, allowing us to approximate them by estimating the return probabilities of random length $\ell$ walks.  Given accurate  estimates of the spectral moments, the spectrum can be recovered by essentially solving the moment-inverse problem, namely recovering a distribution whose moments closely match the estimated spectral moments.

Complementing the above general result, we also give an algorithm with a better dependence on the accuracy parameter $\eps,$ that applies to planar graphs of bounded degree (such as road networks), and generalizations of planar graphs:

\begin{theorem}
\label{theorem:mainP}
For a graph $G$ of maximum degree $d$ that are planar, or that do not contain a forbidden minor, $H$, one can approximate the spectrum of $G$ to earth mover distance $\eps$ in time and queries $\left(\frac{d}{\eps} \right)^{O(\log(1/\eps))}.$
\end{theorem}

The proof of this improved result for bounded degree planar graphs requires two tools.  The first is the observation that the earth mover distance between the spectra of two graphs is at most twice the graph edit distance (the number of edges that must be added/removed to transform one graph into the other).  The second tool is an algorithmic gadget called a ``planar partitioning oracle'' which allows a planar graph of degree at most $d$ to be partitioned into  connected components of size $O(d/\epsilon^2)$, while removing only $\epsilon n$ edges from the graph.   Given such a decomposed graph, the spectrum can then be pieced together from approximations of the spectra of the various pieces.  %We provide the algorithm and proof of Theorem~\ref{theorem:mainP} in the Appendix contained in the Supplementary Material.

We then investigate the consequences of this algorithm for the area of property testing in bounded 
degree graphs. For this purpose we study spectral properties, i.e. properties that
are defined by sets of spectra. We show that for graphs with non-constant girth all $\delta$-robust
spectral properties are testable, i.e. properties where the sets of spectra are not ``thin''.
We believe that this is a first step towards identifying a large class of (constant time) testable 
graph properties that are not hyperfinite.

%\subsection{Our techniques}

The property testing algorithm for testing a $\delta$-robust spectral property $\Pi$ in 
high girth graphs leverages the spectrum estimation algorithm as a subroutine and approximates the distance to 
the set of accepted spectra. If this distance is below a threshold, the algorithm accepts, 
otherwise, it rejects. The difficult part of the analysis is to show that the algorithm rejects 
instances that are $\epsilon$-far. The analysis of this case makes use of a recent result by Fichtenberger
et al. \cite{FPS15} that allows one to construct a small cut between a set $U$ of $\epsilon dn/4$ vertices
and the rest of the graph without changing the distribution of local neighborhoods in the graph.
Since this distribution determines the output distribution of our spectrum approximation algorithm
we also know that the spectrum is not changed much by this operation (something similar can be shown
for the graph $G[V\setminus U]$). We can apply this result to any graph that is accepted by the property
tester, if the spectrum is correctly approximated.
Then we remove all edges incident to $U$ and replace it with a graph whose spectrum is somewhat deep 
inside the set of accepted spectra. This ``moves'' the spectrum of the graph into the set of accepted
spectra. Overall, our construction makes at most $\epsilon dn$ edge modifications and thus the graph
is not $\epsilon$-far.

\subsection{Related work}

Since the 1970's, spectral graph theory has flourished and led to the development and understanding of rich connections between structural and combinatorial properties of graphs, and the eigenvalues and eigenvectors of their associated graph Laplacians (see e.g.~\cite{C97}).  From an algorithmic standpoint spectral methods provide useful tools that have been fruitfully employed to solve a number of graph problems including graph coloring, graph searches (e.g. web search), and image partitioning~\cite{M01,Shi00}.  \iffalse In many of these applications the eigenvectors provide low-dimensional embeddings of the graph, allowing for the application of geometric algorithms. \fi  In terms of the structural interpretations of the eigenvalues, it is easy to see that the multiplicity of the zero eigenvalue is exactly the number of connected components of a graph.  Cheeger's inequality gives a robust analog of this statement, showing a correspondence between the value of the second eigenvalue, and the extent to which the graph can be partitioned into two pieces.  Very recently, a series of works~\cite{KLLOT13,LGT14,LRT12} developed a ``higher order'' Cheeger inequality, quantifying a correspondence between the $i$th eigenvalue  and the extent to which the graph admits a partitioning into $i$ components.

There has been a great deal of work characterizing the spectrum of various models of random graphs, including Erdos-Renyi graphs~\cite{erdHos2013spectral}, and graphs that attempt to model the properties exhibited by real-world graphs and social networks, including random power-law graphs, small-world graphs, and \emph{scale-free} networks (see e.g.~\cite{farkas2001spectra,chung2003spectra}).  One way of testing the plausibility of such models is by comparing their spectrum to those of actual real-world networks, though one challenge is the computational difficulty of computing the spectrum for large graphs, which, in the worst case, requires time cubic in the number of nodes of the graph.

\iffalse In the area of property testing \cite{RS96, GGR98} in the bounded degree graph model \cite{GR02} some
properties related to the spectrum of the (combinatorial) Laplacian such as connectivity \cite{GR02}, bipartiteness\cite{GR99}, expansion \cite{GR00, CS07, KS11, NS10} and $k$-clusterings \cite{CPS15} have 
been studied. When it comes to classes of testable properties, it is known that all minor-closed properties
\cite{BSS08} and, more generally, all hyperfinite properties \cite{NS13} are testable in constant time. 
Hyperfinite properties consist of graphs that can be partitioned into components of constant size by removing
$\epsilon n$ edges and therefore cannot contain expander graphs. Very little is known about properties 
testable in constant time for properties that contain expander graphs. A few examples such as connectivity
or $k$-connectivity can be found in \cite{GR99}, but no non-trivial classes of testable non-hyperfinite properties are known. \fi

Beyond the graph setting, there is a significant body of work from the statistics community on estimating the spectrum of the 
covariance matrix of a high-dimensional distribution, given access to independent samples from the distribution~\cite{K08,LW12}.  As with a graph, 
the eigenvalues of the covariance matrix of a distribution contain meaningful structural information about the distribution in 
question, including quantifying the amount of low-dimensional structure.  Recently,~\cite{KV16} showed that the spectrum of the covariance of a distribution can be accurately recovered given a number of samples that is sublinear in the dimension, by leveraging a method of moments approach that directly estimates the low-order moments of the true spectral distribution.  Although that work is in a rather different setting, we borrow the overall structure, and several technical lemmas, from this moment-based approach.   

\section{Preliminaries}

Let $A$ be an $n \times n$ real-valued matrix. A value $\lambda$ is called an \emph{eigenvalue} of $A$, 
if there exists a vector $v$ such that $Av = \lambda v$. If $A$ is a symmetric matrix then 
its eigenvalues and eigenvectors are real. 
If $A=Q \Lambda Q^{-1}$ where $\Lambda$ is a diagonal matrix, we say that $A$ has an eigendecomposition.
The entries on the diagonal of $\Lambda$ are the eigenvalues and the columns of $Q$ the eigenvectors of
$A$. If $A$ is symmetric and real-valued it always has an \emph{eigendecomposition} of the form $A=Q\Lambda Q^T$, i.e. $Q$ is an 
\emph{orthogonal} matrix ($Q^{-1}= Q^T$).

Two matrices $A$ and $B$ are similar, if they can written as $A= P B P^{-1}$ for an $n \times n$ invertible matrix $P$.
Similar matrices have the same eigenvalues. We may assume w.l.o.g. that the eigenvalues satisfy $\lambda_1 \ge \lambda_2 \ge \dots \ge \lambda_n$, 
(where each eigenvalue appears with its algebraic multiplicity) and refer to this sorted list of eigenvalues as the \emph{spectrum}.  
A matrix is stochastic, if its columns are non-negative reals that sum up to 1.

Throughout, we will also view this list of eigenvalues as a distribution, consisting of $n$ equally-weighted point masses at values $\lambda_1,\ldots,\lambda_n$.  We refer to this distribution as the \emph{normalized spectral measure} or \emph{spectral distribution}.  We will be concerned with recovering this spectral distribution in terms of the \emph{Wasserstein-}$1$ distance metric (i.e. ``earth mover distance'').  We denote the earth mover distance between two real-valued distribution $p$ and $q$ by $W_1(p,q)\,$; this distance represents the minimum, over all schemes of ``moving'' the probability mass of $p$ to yield distribution $q$, where the cost per unit probability mass of moving from probability $x$ to $y$ is $|x-y|$. 

The task of learning the spectral distribution in earth mover distance is closely related to the task of learning the sorted vector of eigenvalues in $\ell_1$ distance.   This is because the $\ell_1$ distance between two sorted vectors of length $n$ is exactly $n$ times the earth mover distance between the corresponding point-mass distributions.  Similarly, given a distribution, $Q$, that is close to the spectral distribution $\mu_{\lambda}$ in Wasserstein distance, one can transform $Q$ into a length $n$ vector whose $\ell_1$ distance is at most $n\cdot W_1(Q,\mu_{\lambda})$.  (See Lemma~\ref{lem:3}.)

In the remainder of this paper we will assume that $A$ is an $n\times n$ real-valued stochastic matrix
with real eigenvalues of absolute value at most $1$ and $n$ linearly independent eigenvectors. 
In particular, we can write $A= Q \Lambda Q^{-1}$.
We use $e_i$ to denote the $i$-th vector of the standard basis of $\mathbb R^n$.

\section{Approximating the spectrum of a stochastic matrix}

In this section we consider the task of approximating the spectrum of a stochastic matrix, $A$, given a certain query access to information about $A$.  Our results on estimating the spectrum of a graph Laplacian, which we give in Section~\ref{sec:graphL}, will follow easily from the results of this section, as learning the spectrum of a graph's Laplacian is equivalent to learning the spectrum of the stochastic matrix corresponding to a random walk on the graph in question.

\subsection{Model of computation}\label{sec:model}

We will assume that we have oracle access to the matrix $A$ of the following form:
On input a number $j$ the oracle provides us with a value $\{1,\dots,n\}$ distributed
according to the $j$-th column of $A$. This type of access to $A$ allows us to perform
a random walk on $A$.   We note that the time it takes to actually implement such an oracle depends on how the graph is represented.  If the graph is stored via adjacency lists then the oracle can be implemented in time $O(d)$ per oracle call; if the neighboring vertices are stored as arrays and the node degrees are also stored, this oracle can be implemented in time constant time per call.

\subsection{Approximating the spectral moments}\label{sec:a1}

We proceed via the method of moments: we first obtain accurate estimates of the low-order moments of the spectral distribution, and then leverage these moments to yield the spectral distribution.

\begin{definition}
Let $A=Q \Lambda Q^{-1}$ be a stochastic $n \times n$ matrix with 
real eigenvalues $1 \ge \lambda_n
\ge \dots \ge \lambda_1 \ge -1$. The $\ell$-th moment of the spectrum of $A$ is defined as $ \frac{1}{n}\sum_{i=1}^n \lambda_i^{\ell}.$
\end{definition}

We will leverage the fact that the trace of a matrix $A$ equals $n$ times the first moment and the trace of 
$A^i =Q \Lambda^i Q^{-1}$ equals $n$ times the $i$-th spectral moment, i.e. $$
\Tr(A^i) = \sum_{i=1}^n \lambda^i.$$

At the same time, we can also view the trace of $A$ as the sum of return probabilities of a
random walk using the transition probabilities of $A$, i.e.
\begin{eqnarray*}
\Tr(A^i) & = & \sum_{j=1}^{n} e_j^T A^i e_j \\
& = & \sum_{1=1}^n \Pr{\text{$i$-step Rand. Walk from $j$
returns to $j$}}.
\end{eqnarray*}
Next we note that we can view 
$$
\frac{1}{n} \dot \sum_{1\le j \le n} \Pr{\text{$i$-step Random Walk from $j$ returns to $j$}}
$$
as the expected return probability of a random walk starting at $j$ when $j$ is chosen 
uniformly at random from $\{1,\dots,n\}$. Thus, given access to $A$ as described in Section~\ref{sec:model}, the following algorithm can be used as 
an unbiased estimator for the spectral moments:

\begin{center}
\myalg{ApproxSpectralMoment($A, \ell, s$)}{
\begin{tabbing}
\hspace{0.5cm}\= {\bf for} $i=1$ {\bf to} $s$ \\
\>\hspace{0.5cm}\= pick $j\in\{1,\dots,n\}$ uniformly at random\\
\>\> $w=j$\\
\>\> {\bf for} $k=1$ {\bf to} $\ell$ {\bf do}\\
\>\>\hspace{0.5cm}\= Let $w'$ be drawn from the distribution of the \\
\>\> $w$-th column of $A$\\
\>\>\> $w = w'$\\
\>\> {\bf if} $w=j$ {\bf then} $X_i =1$ {\bf else} $X_i=0$\\
 \> {\bf return} $\frac{1}{s} \cdot \sum_{i=1}^s X_i$
\end{tabbing}}
\end{center}

The following lemma follows directly from a Hoeffding bound on the sum of independent $0/1$ random variables.
\begin{lemma}\label{lem:1}
Let $s \ge \frac{1}{2} \epsilon^{-2} \ln (2/\delta)$
Given access to the column distributions of a stochastic $n \times n$ matrix $A = Q \Lambda Q^{-1}$
with real eigenvalues $1 \ge \lambda_n,\dots, \ge \lambda_1 \ge -1$,
algorithm {\sc ApproxSpectralMoment}$(A,\ell,s)$ approximates with probability at least 
$1-\delta$ the $\ell$-th spectral moment of $A$ within an additive error $\epsilon$.
The algorithm has a running time of $O(s\ell)$. 
\end{lemma}

\iffalse \begin{proof}
By Hoeffding's bound we get for $0 <\epsilon <1$
$$
\Pr{|\frac{1}{s}\sum_{i=1}^s X_i-\Ex{\frac{1}{s} \sum_{i=1}^s X_i}| \ge \epsilon} \le 2 \cdot e^{-2 \epsilon^2 s} \le \delta .
$$
\end{proof}
\fi

\subsection{Approximating the spectrum from its moments}

In this section we restate results from \cite{KV16} showing that the spectrum can be accurately reconstructed from estimates of the first $\ell$ spectral moments:

\begin{proposition}[Proposition 1 in~\cite{KV16}]
\label{prop1}
Given two distributions with respective density functions $p,q$ supported on $[a, b]$
whose first $k$ moments are $\alpha = (\alpha_1,\dots,\alpha_k)$ and $\beta = (\beta_1,\dots,\beta_k)$, 
respectively, the Wasserstein distance, $W_1(p, q)$, between $p$ and $q$ is bounded by:
$$
W_1(p, q) \le C \cdot \frac{b-a}{k} + g(k) (b-a) \|\alpha - \beta \|_2
$$
where C is an absolute constant, and $g(k) = C'3^k$ for an absolute constant C'.
\end{proposition}

As in~\cite{KV16}, given estimates of the spectral moments, we can recover a distribution whose moments (scaled by a factor of $n$) closely match the estimated moments by solving the natural linear program:

\begin{center}
\myalg{MomentInverse}{\\
\textbf{Inputs:} Vector $\hat{\alpha}$ consisting of the first $\ell$ approximate moments for a distribution supported on the interval $[a,b]$, and a parameter $\eps>0$.\\
\textbf{Output:} Distribution $\bf{p}$.
\begin{enumerate}
\item Define $\bf{x}$ $= x_0,\ldots,x_t$ with $x_i = a + i\eps$ and $t=\lceil \frac{b-a}{\eps} \rceil.$
\item Let $\bf{p^+}$ $=(p^+_0,\ldots,p^+_t)$ be the solution to the following linear program, which should be interpreted as a distribution with mass $p^+_i$ at location $x_i$:
\begin{equation}
\begin{aligned}
\label{eqn:l1obj}& \underset{\bf{p}}{\text{minimize}} & & \|\bf{V}\bf{p}-\hat{\alpha} \|_1 \\ 
& \text{subject to}                        & & \bf{1}^T\bf{p}=1\\
&                                                    & & \bf{p}>0,
\end{aligned}
\end{equation}
where the matrix $\bf{V}$ is defined to have entries $\bf{V}_{i,j} $ $= x_j ^i.$
\item Return distribution $\bf{p^+}.$
\end{enumerate}}
\end{center}

The following lemma leverages Proposition~\ref{prop1} to characterize the performance guarantees of the above algorithm.

\begin{lemma}\label{lem:2}
Consider a distribution $p$ supported on the interval $[a,b]$, and let $\alpha$ denote the vector of its first $\ell$ moments.  Let $\hat{p}$ denote the output of running the MomentInverse algorithm on input $\hat{\alpha},a,b,\eps$.  Then the earthmover distance between $p$ and $\hat{p}$ satisfies: 
$W_1(p,\hat{p}) \le C \frac{b-a}{\ell} + g(\ell)(b-a)\left(\|\alpha-\hat{\alpha}\|_1 + \ell \left((\max(|a|,|b|)+\eps)^\ell-(\max(|a|,|b|))^\ell \right)\right),$\\ where $C$ is an absolute constant and $g(\ell) = O(3^\ell)$ as in Proposition~\ref{prop1}.
\end{lemma}
\begin{proof}
First note that there is a feasible solution to the linear program with objective value at most $\|\alpha-\hat{\alpha}\|_1 + \sum_{i=1}^\ell \left((\max(|a|,|b|)+\eps)^i-(\max(|a|,|b|))^i \right)$ as this is the objective value that would be obtained by discretizing distribution $p$ to be supported at the $\eps$-spaced grid points $x_0,\ldots.$   This quantity hence provides a bound on the $\ell_1$ norm of the difference between the true moments, $\alpha$, and the moments of the distribution returned by the algorithm; since the $\ell_2$ norm is at most the $\ell_1$ norm, this quantity also provides a bound on the $\ell_2$ norm of the discrepancy in moments.  The desired lemma now follows from  applying Proposition~\ref{prop1}.
\end{proof}

\subsection{Approximating the spectrum of $A$}\label{sec:spA}

We now assemble the above components to yield the following theorem characterizing our ability to recover the spectral distribution.

\begin{theorem}
\label{theorem:approx}
Given access to the column distributions of a stochastic $n \times n$ matrix $A=Q \Lambda Q^{-1}$
with real eigenvalues $1\ge \lambda_n,\dots,\lambda_1 \ge -1$, with probability $2/3$
we can approximate the spectrum of $A$ with additive error $\epsilon$ in earth mover distance with running time and query complexity $e^{O(1/\eps)}$.
\end{theorem}

\begin{proof}
The algorithm will  accurately estimate the first $\ell= O(1/\eps)$ spectral moments via Algorithm ApproxSpectralMoment to within accuracy $e^{O(1/\eps)}$
with overall error probability bounded by $1/3$, and then will apply Algorithm MomentInverse to recover a distribution that roughly match\-es the recovered moments.
The proof will follow by assembling Lemmas~\ref{lem:1} and~\ref{lem:2}.  Let the number of spectral moments to estimate be $\ell = c_1/\eps$ for a suitable absolute constant $c_1$, chosen so that the first term in the earth mover bound of Lemma~\ref{lem:2} is at most $\eps/2$.  We will choose the parameter $s$ of Algorithm ApproxSpectralMoment to be $e^{c_2/\eps},$ for a suitable constant $c_2$, so as to guarantee that with probability at least $2/3,$ all $\ell$ spectral moments will be estimated to within error $e^{c_3/\eps},$ where the constant $c_3$ is selected so that the bound from the $\|\alpha-\hat{\alpha}\|$ portion of the second term is at most $\eps/4.$  Finally, the discretization parameter in the support of the linear program of MomentInverse will be chosen to be $e^{c_4/\eps}$, for a constant $c_4$ so as to ensure that the contribution from the final portion of the bound of Lemma~\ref{lem:2} is also bounded by $\eps/4$.  
\end{proof}

While the MomentInverse algorithm returns a distribution $\hat{p}$ described via $e^{O(1/\eps)}$ numbers, we note that there is a simple algorithm, computable in $O(n\,e^{O(1/\eps)})$ time, that will convert $\hat{p}$ into a vector $v$ of length $n$, with the property that the earth mover distance between the spectral distribution $p$ and the distribution associated with $v$ (consisting of $n$ equally-weighted point masses at the locations specified by $v$) is at most the distance between $p$ and $\hat{p}$.

\begin{center}
\myalg{DiscretizeSpectrum($n,\bf{q}$)}{
\\ \textbf{Input:} Distribution $\bf{q}$ consisting of a finite number of point masses, integer $n$.\\
\textbf{Output:} Vector $\bf{v}$ $=(v_1,\ldots,v_n).$
\begin{enumerate}
\item Let $f_q:[0,1] \rightarrow \mathbb{R}$ be defined to be the non-decreasing function with the property that for $X$ drawn uniformly at random from the interval $[0,1]$, the distribution of $f_q(X)$ is $\bf{q}.$
\item Set $$v_i = \E \left[ f_q(X)| X \in [\frac{i-1}{n},\frac{i}{n}] \right],$$ and return $\bf{v}$ $=(v_1,\ldots,v_n).$
\end{enumerate}}
\end{center}
\iffalse \footnote{Why does $f_q$ always exist?}\fi

\begin{lemma}\label{lem:3}
Consider a distribution $p$ that consists of $n$ equally weighted point masses.  Let $q$ be any distribution consisting of a finite number of point masses, and let $q_v$ denote the distribution consisting of $n$ equally weighted point masses located at the values specified by the vector $v$ returned by running Algorithm DiscretizeSpectrum on inputs $n$ and $q$.  Then the earth mover distance between $p$ and $q_v$ satisfies $$W_1(p,q_v) \le W_1(p,q).$$
\end{lemma}
\begin{proof}
Let $p_1,\ldots,p_n$ with $p_i \le p_{i+1}$ denote the support of distribution $p$.  Observe that the earth moving scheme of minimal cost that yields distribution $p$ from distribution $q$ consists of moving the $1/n$ probability mass in distribution $q$ corresponding to the (scaled) conditional distribution $f_q(X)$ conditioned on $X \in [\frac{i-1}{n},\frac{i}{n}]$ to location $p_i$.  Let $q_i$ denote the $i$th such conditional distribution.  Since, $W_1(p,q_v) = \frac{1}{n} \sum_{i=1}^n |p_i - v_i|,$  it suffices to analyze $|p_i-v_i|$ independently for each $i$.   To conclude, note that the contribution of $q_i$ to the earthmover distance is simply 
\begin{eqnarray*}
\frac{1}{n} \sum_{x \in supp(q_i)} |x-p_i|\cdot q_i(x) &\ge& \frac{1}{n}\left|p_i - \sum_{x \in supp(q_i)} x \cdot q_i(x) \right| \\
&=& \frac{1}{n}|p_i-v_i|,
\end{eqnarray*}
 where for $x\in supp(q_i),$ we use the shorthand $q_i(x)$ to denote the amount of mass that distribution $q_i$ places on value $x$.
\end{proof}

\section{Approximating the spectrum of graph Laplacians}\label{sec:graphL}

In this section we describe how to leverage the results of Section~\ref{sec:spA}, namely how to accurately approximate the spectrum of a stochastic matrix, to recover the spectrum of a graph Laplacian.
Let $G=(V,E)$, $V=\{1,\dots,n\}$ be an undirected graph and let $A$ be 
its adjacency matrix. We assume that we have access to an oracle that 
on input a vertex $v$ can provide a uniformly distributed neighbor of $v$.

\begin{definition}\label{def:normL}
The \emph{normalized Laplacian} of a graph $G$ with adjacency matrix $A$
is defined as $L_G = I - D^{-1/2} A D^{-1/2}$, where $D$ is a diagonal matrix 
whose entries are the vertex degrees.
\end{definition}

Let $M=AD^{-1}$ be the transition matrix of a random walk on $G$, i.e. 
$M_{i,j} = \frac{1}{\deg(j)}$ whenever there is an edge between vertex $i$ and
$j$ and where $\deg(j)$ denotes the degree of vertex $j$.
Note that $M=D^{1/2} (I - L_G) D^{-1/2}$ and so $M$ is similar to the real
valued symmetric matrix $I-L_G$. Thus, $M$ is a stochastic matrix that can be written 
as $Q\Lambda Q^{-1}$ and the $i$-th largest eigenvalue $\lambda_i$ of $L_G$ corresponds to an $i$-th smallest 
eigenvalue $1-\lambda_i$ of $M$ (in particular, the eigenvalues are real).

Hence approximating the spectrum of $M$ will also give an approximation
of the spectrum of $L_G$, immediately yielding Theorem~\ref{theorem:mainL}.

\iffalse
\begin{corollary}
Let $G=(V,E)$ be an undirected graph and let us assume that we can query for
random uniformly distributed neighbors in $O(1)$ time. Then with probability at least $2/3$ 
one can approximate the spectrum of the normalized Laplacian of $G$ with additive error $\eps$ 
in earth mover distance in $exp(O(1/\eps))$ time. 
\end{corollary}\fi

\section{An Improved Algorithm for Bounded Degree Planar Graphs}

In this section we describe an improved algorithm for bounded degree planar graphs and, more generally, 
minor-closed bounded-degree graphs, establishing Theorem~\ref{sec:spA}. We need two main tools to obtain this result. The first 
one is a lemma that shows that the earth mover distance is at most twice the graph edit distance.

\begin{lemma}
Let $G=(V,E)$ and $H=(V,E')$ be two graphs. Then
$$
|V| \cdot  W_1(\lambda, \lambda') \le 2 G \Delta H , 
$$
where $G \Delta H$ denotes the number of edges that need to be changed to transform $G$ into an isomorphic
copy of $H$ and $\lambda$ and $\lambda'$ are the spectra of $G$ and $H$, respectively.
\end{lemma} 

\begin{proof}
We first recall the variational characterization of eigenvalues for a symmetric $n\times n$ matrix $Q$:
$$\lambda_k(Q) = \min\{\lambda| \exists\mathrm{k-subspace}\; F\subset \R^n \;\mathrm{ s.t. } \;\forall x\in F\;x^tQx \leq \lambda x^tx\}$$
Let $U\subset \R^V$ be the subspace of functions that vanish on the vertices incident to at least an edge that is in one of the graphs $G$ and $H$ only. By assumption, the codimension of $U$ is at most $G\Delta H$. Now, it is easy to see that the (normalized) Laplacian quadratic forms $L_G$ and $L_H$ coincide on $U$. For $0\leq \lambda\leq 2$, let $f_G(\lambda)$ (resp. $f_H(\lambda)$) be the fraction of eigenvalues of $G$ (resp. $H$) that are below $\lambda$. From the variational principle, for a given $\lambda$, there is a $f_G(\lambda)|V|$-subspace $W\subset R^V$ such that $\forall x\in W\, x^tL_G x\leq \lambda x^tx$. The subspace $W\cap U$ is at least $f_G(\lambda)|V|-G\Delta H$ dimensional and because the two quadratic forms coincide on it, it witnesses that $f_H(\lambda)|V|\geq f_G(\lambda)|V|-G\Delta H$ using the variational principle. By symmetry, $|f_G-f_H|\leq G\Delta H /|V|$. 

Since $f_G$ and $f_H$ coincide outside $[0,2]$, we see that $\int|f_G-f_H| \leq 2G\Delta H/|V|$. The latter integral is the area between the graphs of $f_H$ and $f_G$. Now, switching axes, these graphs become the graphs of the inverse cumulative distribution functions of the spectral measures of $G$ and $H$. Since the earth mover distance is the $L_1$ distance between inverse cumulative distribution functions, the result follows.
\end{proof}

The second tool is an algorithmic gadget called a ``planar partitioning oracle''. It is well known that by
applying the planar separator theorem \cite{LT79} multiple times one can partition a planar graph
with maximum degree $d$ into connected components of size $O(d/\epsilon^2)$ by removing 
$\epsilon n$ edges from the graph. A planar partitioning oracle provides local access to such a partition.

\begin{definition}[\cite{HKNO09} ] 
We say that $\mathcal{O}$ is an $(\epsilon,k)$-partitioning oracle for a class $C$ of graphs 
if given query access to a graph $G = (V,E)$ in the adjacency-list model, it provides query access to a partition $P$ of $V$. 
For a query about $v\in V$ , $\mathcal O$ returns $P[v]$. The partition has the following properties: 
\begin{itemize}
\item
$P$ is a function of the graph and random bits of the oracle. In particular, it does not depend on the order of queries to 
$\mathcal O$. 
\item For every $v \in V,$ $|P[v]|\leɠk$ and $P[v]$ induces a connected graph in $G$. 
\item
If $G$ belongs to $C$, then $|\{(v,w) \in E : P[v] \neq P[w]\}| \le ʜepsilon|V|$ with probability $9/10$.
\end{itemize}
\end{definition}

We will leverage a partitioning oracle by Levi and Ron:

\begin{theorem} [\cite{LR15}] 
For any fixed graph $H$ there exists an $(\epsilon,O(d^2/\epsilon^2))$-partition-oracle for $H$-minor free graphs that makes 
$(d/\epsilon)^{O(log(1/\epsilon))}$ queries to the graph for each query to the oracle. The total time complexity of a sequence 
of $q$ queries to the oracle is $q \log q \cdot (d/\epsilon)^{O(\log(1/\epsilon))}$. 
\end{theorem}

The partitioning oracle provides us access to a partition of a minor-closed graph $G=(V,E)$ into small connected components. This
partition is obtained by removing at most $\epsilon n$ edges. Let us call the graph that consists of these connected components 
$H(V,E')$. By our first lemma the spectra of $G$ and $H$ have earth mover distance at most $\epsilon$. This means that if we can
approximate the spectrum of a graph with small connected components, then we can also estimate the spectrum of a minor-closed
bounded degree graph using the partitioning oracle from above. 

We now provide a simple algorithm that samples eigenvalues from the spectrum of a graph with small connected components.

\begin{center}
\myalg{SmallCCSpectrum($H=(V,E)$)}{
\\ \textbf{Input:} Graph $H=(V,E)$ with small connected components.\\
\textbf{Output:} A random eigenvalue of the normalized Laplacian of $H.$
\begin{enumerate}
\item Sample a vertex $v\in V$ uniformly at random \\
\item Compute the connected component $C(v)$ of $v$\\
\item Return a random eigenvalue of the normalized Laplacian of $C(v)$ \\
\end{enumerate}}
\end{center}

\begin{lemma}
Algorithm SmallCCSpectrum samples a random eigenvalue from $H$. If all connected components are of 
size at most $t$ then the running time of the algorithm is $O(t^3)$.
\end{lemma}

\begin{proof}
First we observe that the spectrum of $H$ is the union of the spectrum of its connected components. Indeed, 
given an eigenvalue with corresponding eigenvector of a connected component $C(v)$ of $H$ we observe that 
extending the eigenvector with $0$ will yield an eigenvector of $H$ with the same eigenvalue.

Next we observe that the algorithm returns a uniformly distributed eigenvalue of $H$. Let us fix an eigenvalue
$\lambda_i$ belonging to connected component $C$. The probability to sample $\lambda_i$ is the probability to
sample a vertex from the connected component (which is $|V(C)|/|V|$) times the probability that $\lambda_i$ 
is sampled from the $|V(C)|$ eigenvalues of the connected component, which is $1/|V(C)|$. Hence the probability to sample
$\lambda_i$ is $1/|V|$.
\end{proof}

\medskip
\noindent \textbf{Theorem~\ref{theorem:mainP}.}\emph{
Let $\mathcal G$ be a family of graphs of maximum degree at  most $d$ that does not contain a forbidden minor $H$. Then one can approximate the
spectrum of $G$ in earth mover distance upto an additive error of $\epsilon$ in time $\left(\frac{d}{\epsilon}\right)^{O(\log (1/\epsilon))}$.}

\begin{proof}
The approximation guarantee follows from the relation between edit distance and earth mover distance and when we estimate the spectrum
using polynomially (in $1/\epsilon$) many calls to algorithm SmallCCSpectrum. The running time then follows from the running time of 
the planar partitioning oracle (where the additional factors in $1/\epsilon$ are absorbed by the $O$-notation in the exponent). 
\end{proof}

\section{Testing Spectral Properties}

In this section we study the implications of our result on the area of property
testing in the bounded degree graph model. We start by giving some basic definitions.
We will consider the bounded-degree graph model introduced by Goldreich and Ron \cite{GR02}. 
In this model the degree of a graph is bounded by $d$, which we typically think of being
a constant although we will parametrize our analysis in terms of $d$. A graph 
with maximum degree bounded by $d$ is also called $d$-degree bounded graph.
 We assume that the input graph $G=(V,E)$ has vertex set $V=\{1,\dots,n\}$ and $n$ is given to the 
algorithm. In the bounded degree graph model we can query for the $i$-th neighbor
adjacent to vertex $j$. If no such vertex exists, the answer to the query is a special symbol 
indicating this.

The goal of property testing is to study a relaxed decision problem for graph properties, where
a graph property is defined as follows:

\begin{definition}
A graph property $\Pi$ is a set of graphs that is closed under isomorphism.
For a graph property $\Pi$ we use $\Pi_n$ to denote the subset of graphs in $\Pi$
that have exactly $n$ vertices.
\end{definition}

In this relaxed decision problem we are studying how to approximately decide whether
an input graph has a given graph property $\Pi$ or is far away from $P$. A property testing algorithm
for property P (also called property tester) is given
access to an input graph in the way described above and it has to accept with probability
at least $2/3$ every input graph that has property $\Pi$ and has to reject with probability
at least $2/3$ every input graph that is $\epsilon$-far from $\Pi$ according to the
following definition.

\begin{definition}
A $d$-bounded degree graph $G=(V,E)$  is $\epsilon$-far from a property $\Pi$, 
if one has to insert/delete  more than $\epsilon d |V|$ edges in $G$ to obtain a 
$d$-bounded degree graph that has property $\Pi$. 
\end{definition}

One of the main questions studied in the area of property testing in bounded degree graphs
is to identify the properties that area testable in constant time, for example, 
according to the following definition.

\begin{definition}
A graph property $\Pi$ is \emph{testable} in the bounded degree graph model with degree bound $d$, 
if there exists a function $q(\epsilon,d)$ such that for every $\epsilon>0$, $d,n \in \mathbb N$ there exists an 
algorithm $A_{\epsilon,d,n}$ such that 
\begin{itemize}
\item 
$A_{\epsilon,d,n}$ makes at most $q(\epsilon,d)$ queries to the graph,
\item
$A_{\epsilon,d,n}$ accepts with probability at least $2/3$ every $d$-bounded 
degree graph $G\in \Pi_n$ 
\item
$A_{\epsilon,d,n}$ rejects with probability at least $2/3$ every $d$-bounded 
degree graph that is $\epsilon$-far from $\Pi_n$.
\end{itemize}
\end{definition}

It is known that some fundamental graph properties like connectivity, $c$-vertex connectivity and $c$-edge connectivity
are testable \cite{GR02}. Also, properties like subgraph-freeness or some properties that depend on the distribution of vertex degrees are (trivially) testable. Furthermore, it is known that all minor-closed properties \cite{BSS08} and, more generally, 
all hyperfinite properties are testable \cite{NS13}, where a property is hyperfinite, if all graphs that have the property
can be decomposed into small components by removing $\epsilon dn$ edges from the graph. Thus, hyperfinite graphs can
be thought of as the opposite of expander graphs, for which small cuts do not exist. Not much is known about (constant time) testable
properties of expander graphs or properties that contain expander graphs except for the properties mentioned above.
Our result indicates that some properties that depend on the spectrum of the graph may be testable and in this section
we initiate the study of such properties. We then prove that a certain class of spectral properties is testable for any class of high girth graphs, i.e. when the input graph is promised to have high girth. 
In the following we will  view the spectrum as an $n$-dimensional vector 
$\Lambda = (\lambda_1,\dots, \lambda_n)$. We will also sometimes refer to the $l_1$-distance between two spectra $\Lambda_1$, $\Lambda_2$ (viewing them as sorted vectors) 
which is equals the earth mover (or Wasserstein) distance of the 
corresponding spectral measures, scaled by a factor of $n$, i.e. $\frac{1}{n}\|\Lambda_1-\Lambda_2\|_1 = W_1(\Lambda_1,\Lambda_2)$.
We start with a definition of spectral graph properties.

\begin{definition}
A graph property $\Pi$ of $d$-bounded degree graphs is called \emph{spectral}, 
if for every $n \in \mathbb N$ 
there exists a set $S_n \subseteq \mathfrak X_n$ such that $\Pi_n$ is the set of all
$d$-bounded degree graphs on $n$ vertices whose spectrum is in $S_n$. Here, 
$\mathfrak X_n \subseteq[0,2]^n$ is the set of spectra that are realized by $d$-bounded degree graphs with $n$ vertices.
%Given graph property $\Pi^*$ we say that a property $\Pi$ is a spectral property
%of $\Pi^*$, if for every $n \in \mathbb N$ 
%there exists a set $S_n \subseteq [0,2]^n$ such that $\Pi_n$ is the set of all graphs from $\Pi^*$
%graphs whose spectrum is in $S_n$. 
\end{definition}

We would like to use our algorithm from the previous section as a property tester. 
The rough idea is that we would like to accept all graphs whose spectrum is close 
(in $l_1$-distance) to the set $S_n$. The technical difficulty is to relate the 
edit distance between graphs to the distance between their spectra.

In order to prove that all spectral graph properties are testable, it would suffice to 
prove a statement similar to the following: If $G$ is $\epsilon$-far from $\Pi_n$ then the 
$l_1$-distance of the spectrum of $G$ to $S_n$ is at least $\eta(\epsilon) n$ for some $\eta(\epsilon)>0$.
However, we do not believe that such a general statement is true. Therefore, we restrict our
attention to the following class of properties:
\begin{definition}
Let $\mathfrak G$ be a class of graphs that is closed under isomorphism.
A graph property $\Pi$ of $d$-bounded degree graphs is called \emph{$\delta$-robustly spectral}, if for every $n \in \mathbb N$ 
there exists a set $S_n \subseteq \mathfrak X_n$ such that $\Pi_n$ is the set of all $n$ vertex graphs in $\mathfrak G$
whose spectrum has $l_1$-distance at most $\delta n$ to $S_n$. Here, $\mathfrak X_n$ is the set of spectra that are realized by $d$-bounded degree graphs in $\mathfrak G$ with $n$ vertices.
%Given graph property $\Pi^*$ we say that a property $\Pi$ is a \emph{$\delta$-robustly spectral} property of $\Pi^*$, if for every $n \in \mathbb N$ 
%there exists a set $S_n \subseteq \mathfrak X_n$ such that $\Pi_n$ is the set of all graphs from $\Pi^*$
%graphs whose spectrum has $l_1$-distance at most $\delta$  to $S_n$.
\end{definition}

In the following we will consider $\mathfrak G$ to be a class of high girth graphs according to the following
definition.

\begin{definition}
 A class of graphs $\mathfrak G$ has \emph{high girth}, if there exists $f(n) = \omega(1)$ such that
every $n$ vertex graph in $\mathfrak G$ has girth at least $f(n)$.
\end{definition}

The main result in this section is the following theorem.

\begin{theorem}\label{thm:spectralgraphprop}
Every $\delta$-robust property is testable in the bounded degree graph model when the input is restricted to a class of high girth graphs.
\end{theorem}
\begin{proof}[Proof of Theorem~\ref{thm:spectralgraphprop}]
Let $d>0$ and let $\mathfrak G$ be a class of high girth graphs with maximum degree bounded by $d$.
Let $\delta >0$ be given and let $\Pi$ be a $\delta$-robust property for the class $\mathfrak G$ with the sets
$S_n$ be as in the definition above. 
We need to show that for every $\epsilon>0$ and $n \ge 0$ there is an an algorithm $A_{\epsilon,d,n}$ that accepts with probability at least $2/3$ 
every $n$ vertex graph from $\Pi$ and that rejects with probability at least $2/3$ every $n$-vertex graph from $\mathfrak G$ that is $\epsilon$-far from $\Pi$. 

Thus let us fix an arbitrary $1\ge \epsilon >0$ and $n \ge 0$.
We will assume that $n \ge N_0=  N_0(\delta,\epsilon, d)$ for sufficiently large $N_0$. We will also need to define $\epsilon^* = \epsilon^*(\delta,\epsilon) = \epsilon \delta /16$. The algorithm will be as follows:

\begin{tabbing}
{\sc TestRobustlySpectral} $(G)$\\
\hspace{0.5cm} \= {\bf if} $n < N_0$ {\bf then} query all edges of $G$ and accept, iff $G\in \Pi$\\
\> {\bf else} \= Let $\Lambda$ be an approximation of the spectrum of $G$ with error at most $\epsilon^*$\\
\>\> {\bf if} $\inf_{q\in S_n} ||\Lambda-q||_1 \le \delta + \epsilon^*$ {\bf then} accept\\
\>\> {\bf else} reject\\
\end{tabbing}

We will first argue that the algorithm always accepts, if $G \in \Pi$. Indeed, if $n<N_0$ 
we accept, iff $G$ is in $\Pi$. If $n \ge N_0$ and the output $\Lambda$ of our spectrum approximation algorithm
is a approximation of the true spectrum $\Lambda^*$ of $G$ with additive error at most $\epsilon^* n$ (which happens with probability at least $2/3$), then we know that 
$\inf_{q\in S_n} \|\Lambda^*-q\|_1 \le \delta$ by the definition of $\delta$-robust and
$\|\Lambda - \Lambda^*\| \le \epsilon^*$ by the properties of the approximation
algorithm. By the triangle inequality we get 
$$\min_{q\in S_n} \|\Lambda-q\|_1 \le \min_{q\in S_n} \|\Lambda^*-q\|_1 + \|\Lambda^*-\Lambda\|_1
\le  \delta + \epsilon^* .$$
Hence, the algorithm accepts $G \in \Pi$ with probability at least $2/3$.

It remains to prove that any graph $G$ that is $\epsilon$-far from $\Pi$ is rejected with probability at least $2/3$. 
We first observe that every graph whose spectrum has distance more than $\delta + 2\epsilon^*$ to $S_n$
will be rejected with probability at least $2/3$. We prove that all graphs whose spectrum has distance at
most $\delta+ 2\epsilon^*$ to $S_n$ are indeed $\epsilon$-close to $\Pi$.
This is done in the following lemma.

\begin{lemma}
\label{lemma:far}
Let $0 < \epsilon,\delta,1$ and let $\epsilon^* \le \epsilon \delta/16$.
Let $G\in \mathfrak G$ be a high girth $d$-bounded degree graph whose spectrum has 
distance at most $\delta + 2 \epsilon^*$ to $S_n$.
Then we can modify at most $\epsilon d n$ edges of $G$ to obtain a graph $G^*$ whose spectrum
has distance at most $\delta$ to $S_n$.
\end{lemma}

\begin{proof}
Let $G$ be as in the lemma and let $\Lambda$ be the spectrum of $G$. 
The proof consists of two steps. First we show that we can modify $\epsilon d n/2 $ 
edges of $G$ to obtain a graph $G'$ that has a small cut between a set $V_1$ of size 
$\epsilon n/4$ and the rest of the graph and that has the same frequencies of 
local neighborhoods as $G$. Furthermore, the frequencies of local neighborhoods
of the vertices in $V_1$, and in the complement $V_2$, respectively, is also approximately the same
as in $G$.

Then we remove all edges incident to $V_1$ to obtain a graph $G'$ and we define $H_2 = G'[V_2]$.
Since the cut between $V_1$ and $V_2$ is small, this does not change too many local
neighborhoods in $V_2$ and the frequencies of local neighborhoods are still an approximation
of the frequencies in $G$. Since the output distribution of our algorithm {\sc ApproximateSpectrum}
is also fully determined by the frequencies of local neighborhoods, this also implies
that the spectrum of $H_2$ is an approximation of the spectrum of $G$.

We then replace $G'[V_1]$ by a graph $H_1$ on vertex set $V_1$ that has approximately 
the spectrum $\arg \min_{q\in S_n} \|\Lambda - q\|_1$. Let $H \in \Pi_n$ denote an $n$-vertex graph with
spectrum $\arg \min_{q\in S_n} \|\Lambda - q\|_1$.
The existence of such a
graph $H_1$ is proven below. Finally, we argue that the new graph has distance at 
most $\delta$ to $S_n$. 

We start with our first modification of turning $G$ into $G'$. This is done using the following 
lemma from \cite{FPS15} (here $V_1 \times V_2$ refers to the set of undirected pairs). We need the
following notation. A $k$-disc $\disc_k(G,v)$ is the subgraph that is induced by all vertices
of distance at most $k$ to $v$ and that is rooted at $v$. We say that two $k$-discs $\Gamma$ and $\Gamma'$ are isomorphic
if there is a graph isomorphism between them that maps the root of $\Gamma$ to the root of $\Gamma'$.
We write $\Gamma \simeq \Gamma'$ in that case. We denote the number of isomorphism classes of
$k$-discs of $d$-bounded degree graphs as $L=L(d,k)$ and denote $\mathcal T_k = \{T_1,\dots,T_L\}$
to be a corresponding set of graphs, i.e. the graphs are pairwise non-isomorphic.
We use $\freq_k(G)$ to be an $L$-dimensional vector such that the $i$-ith entry denotes the
fraction of $k$-discs in $G$ that are isomorphic to $T_i$. This vector describes the distribution
of local neighborhoods in $G$. We write $\freq_k(U|G)$ to denote an $L$-dimensional vector such that
its $i$-th entry is the fraction of the $k$-discs rooted at the vertices in $U$ that are isomorphic 
to $T_i$.

The lemma from \cite{FPS15} quantifies
the observation that in a graph with girth at least $2k+2$ with four vertices $a,b,c,d$ such that
$(a,b), (c,d) \in E$ and such that the $k$-disc type of $a$ equals the $k$-disc type of $c$ and
the $k$-disc type of $b$ equals that of $d$ we can replace edges $(a,b), (c,d)$ by $(a,d)$ and $(b,d)$ 
without changing the $k$-disc types of any vertex provided that the distance from $a$ to $c$ and
$b$ to $d$ is sufficiently large. 

\begin{lemma}
\label{Lemma:Partition}
Let $G = (V,E)$ be a $d$-bounded graph with girth$(G) \ge 2k + 2, k \in \mathbb N, \eta \in [0, 1]$ 
and let $V_1 \cup V_2 = V$, $V_1\cap V_2 = \emptyset$, be a partitioning of $V$ such that 
$|\freq_k(V_1 | G)_{\Gamma} - \freq_k(V_2 | G)_{\Gamma} | \le \eta$
for all $k$-discs $\Gamma \in \mathcal T_k$. Then either there exists a graph $H = (V, F)$ such that 
\begin{itemize}
\item[(1)]
girth$(H) \ge 2k + 2$ 
\item[(2)]
$|F \ (V_1 \times V_2)| \le |E \ (V_1 \times V_2)| - 2$ 
\item[(3)]
$ \disc_k(H,w) \simeq \disc_k(G,w) \forall w \in V$
\end{itemize}
or the cut between $V_1$ and $V_2$ is small:
$$
e(V_1, V_2)\le 6d^{2k+2}L + \eta Ld \cdot min(|V_1|, |V_2|) .
$$
\end{lemma}

Now let $k=k(\epsilon^*,d)$ be the length of the random walks performed by {\sc ApproximateSpectrum} 
on input parameter $\epsilon^*$ and let $s=s(\epsilon^*,d)$ be the number of vertices sampled
uniformly at random by the algorithm. Since the family of graphs we consider has high girth, 
we know that for sufficiently large $n$ all graphs have girth at least $2k+2$.
Since the random walks performed by {\sc ApproximateSpectrum} are of length at most $k$, the output distribution of algorithm {\sc ApproximateSpectrum} is fully determined by the 
distribution of $k$-discs in the input graph $G$, i.e. $\freq_k(G)$. 
%Let us set $\eta=\frac{1}{d^{k+1}} \cdot 1/(2Ld)$. 
We partition $V$ into 
sets $V_1$ and $V_2$ such that $\|\freq_k(V_2|G) - \freq_k(G)\|_1 \le \frac{1}{20s}$.
Clearly, such a partition exists for sufficiently large $n$.
Then we apply Lemma \ref{Lemma:Partition} repeatedly until we obtain a small cut. 
Since $|V_1| = \epsilon n/4$ and since in each iteration we do $4$ edge modifications
to decrease the cut size by $2$, we modify at most $\epsilon d n/2$ edges in this way.
We end up with a cut that satisfies the small cut condition of the lemma.
Now we observe that removing an edge can change at most $d^{k+1}$ $k$-discs.
We observe that for sufficiently large $n$ we get
$$
d^{k+1} \cdot e(V_1,V_2) \le \frac{n}{80s} 
.
$$
Thus, we can remove all edges incident to $V_1$ to obtain a graph $G'$ that satisfies
$$
\|\freq_k(V_2|G') - \freq_k(G)\|_1 \le \|\freq_k(V_2|G') - \freq_k(V_2| G)\|_1 + \|\freq_k(V_2|G) - \freq_k(G)\|_1 
\le 1/(10s).
$$
We now apply the lemma below on $G$ and $H_2 = G'[V_2]$ to obtain that
$W_1(\Lambda_G,\Lambda_{H_2}) \le 2 \epsilon^*$.

\begin{lemma}
\label{lemma:CloseDiscs}
Let $G, H$ be two $d$-bounded degree graphs. 
Let $\epsilon \in[0,1]$, $d\ge 1$. Let $s=s(\epsilon^*,d)$ be the number of
vertices sampled uniformly at random by algorithm {\sc ApproximateSpectrum} with input parameter $\epsilon^*$.
Let $k=k(\epsilon^*,d)$ be the length of the random walks performed by the algorithm.
If $\|\freq_k(G)-\freq_k(H)\|_1 \le 10/s$ then $W_1(\Lambda_G,\Lambda_H) \le 2\epsilon^*$.
\end{lemma}

\begin{proof}
We consider algorithm {\sc ApproximateSpectrum} with input parameter $\epsilon^*$ on input $G$ and
$H$ respectively. We observe that the output distribution of the algorithm is fully determined $\freq_k(G)$
and $\freq_k(H)$, respectively. Since our algorithm samples $s$ vertices uniformly at random this implies that the probability that our algorithm on input $G$ and $H$ behave differently is at most $1/10$. 
This implies that there exists an output $\widetilde \Lambda$, which is guaranteed to be an
additive $\epsilon^*$ approximation for $\Lambda_G$ and $\Lambda_H$. By the triangle inequality we obtain
$W_1(\Lambda_G,\Lambda_H) \le 2\epsilon^*$.
\end{proof}

Next we will construct the graph $H_1$. We need the following lemma to control the spectrum of the 
union of two disjoint graphs.

\begin{lemma}
\label{Lemma:Composition}
Let $G_1=(V_1,E_1)$ be a graph with $k$ vertices and let $G_2=(V_2,E_2)$ be a graph
with $\ell$ vertices, $V_1 \cap V_2 = \emptyset$. Let $\lambda_1,\dots,\lambda_k$ be the eigenvalues of the Laplacian of $G_1$ and
$\nu_{1},\dots\nu_{\ell}$ be the eigenvalues of the Laplacian of $G_2$. Then
$\lambda_1,\dots,\lambda_k, \nu_{1},\dots\nu_{\ell}$ are the eigenvalues of the Laplacian of $G^* = (V_1\cup V_2, E_1 \cup E_2)$.
\end{lemma}

\begin{proof}
It is easy to verify that the eigenvectors of the Laplacian of $G^*=(V_1\cup V_2,E_1 \cup E_2)$ are the eigenvectors
of the Laplacians of $G_1$ and $G_2$ filled up with zeros. The result follows immediately by observing that the corresponding
eigenvalues do not change.
\end{proof}

We then use the Claim below to construct $H_1$ from our graph $H$ with spectrum $\arg \min_{q\in S_n} \|\Lambda - q\|_1$.

\begin{claim}
There exists $N_1=N_1(\epsilon, \delta, d)$ such that for every $d$-bounded degree graph $H=(V,E)$
with $|V| \ge N_1$ with spectrum $\Lambda_H$ there is a $d$-bounded degree graph $H_1$ on $\epsilon |V|/4$
vertices such that $W_1(\Lambda_H,\Lambda_{H_1}) \le \epsilon^*$.
\end{claim}

\begin{proof}
Since the set of all spectra has an $\epsilon^*$-net with respect to the Wasserstein distance whose
size does only depend on $\epsilon^*$, we obtain that for every $d\ge 1$ the size of the
smallest graph whose spectrum has Wasserstein distance at most $\epsilon^*$ to the spectrum
of $H$ is a function of $\epsilon^*$ and $d$. 
In particular, there exists a graph $H'$ of size depending only on $\epsilon^*/2$ and $d$ with 
$W_1(\Lambda_H,\Lambda_{H'})\le \epsilon^*/2$. For sufficiently large $n$ 
we can now define $H_1$ to be the union of $\epsilon |V|/|V(H')|$ 
copies of $H'$ plus isolated vertices. By Lemma \ref{Lemma:Composition} we obtain the bound on the 
spectrum for sufficiently large $n$, i.e. we can define $N_1=N_1(\epsilon,\delta,d)$ such
that the bound on the spectrum holds for every $n \ge N_1$.
\end{proof}

Thus, our construction yields two graphs $H_1=(V_1,E_1)$ and $H_2=(V_2,E_2)$ such
that $W_1(\Lambda_{H_1}, \Lambda_H) \le \epsilon^*$ and $W_1(\Lambda_{H_2}, \Lambda_G) \le 2\epsilon^*$.
We can now finish the proof of our lemma by showing that $G^*=(V_1\cup V_2,E_1\cup E_2)$ has
Wasserstein distance at most $\delta$ to $H$ (and hence $l_1$-distance to $S_n$).
We obtain that
$
W_1(\Lambda_{G^*}, \Lambda_H) \le \frac{\epsilon}{4} \cdot W_1(\Lambda_{H_1},\Lambda_H) + (1-\frac{\epsilon}{4}) W_1(\Lambda_{H_2},\Lambda_H) 
\le \frac{\epsilon}{4} \cdot \epsilon^* + (1-\frac{\epsilon}{4}) (W_1(\Lambda_{H_2},\Lambda_G) + W_1(\Lambda_G,\Lambda_H))
\le (1-\frac{\epsilon}{4}) (2 \epsilon^* + \delta + 2 \epsilon^*) + \frac{\epsilon}{4} \cdot \epsilon^*
\le \delta - \frac{\epsilon \delta}{4}  + 4 \epsilon^*
\le \delta
$
for our choice of $\epsilon^* \le \epsilon \delta/16$.
\end{proof}
This also finishes the proof of our main theorem.
\end{proof}

\section{Experiments}
In this section we demonstrate the practical viability of our spectrum estimation approach.  We considered 15 undirected network data\-sets that are publicly available on the Stanford Large Network Dataset Collection~\cite{snapnets}.   These data\-sets include three road networks (ranging from 1M nodes to 1.9M nodes), six co-authorship networks including DBLP collaboration network (317k nodes, 1M edges), and six social networks including small portions of Facebook (4k nodes, 88k edges), Twitter (81k nodes, 1.7M edges), and Google+ (107k nodes, 13M edges), as well as the LiveJournal social graph, (4M nodes, 34M edges), Orkut (3M nodes, 117M edges), and a portion of the Youtube user follower graph (1M nodes, 2.9M edges).

All experiments were run in Matlab on a MacBook Pro laptop, using Matlab's \emph{graph} datastructure to store the networks.   For each network, we ran our spectrum estimation algorithm 20 times and then averaged the 20 returned spectra.   Each of the spectra was obtained by simulating 10k independent random walks of length 20 steps each, and then leveraging our \textsc{ApproxSpectralMoment} algorithm of Section~\ref{sec:a1} to estimate the first 20 spectral moments.  These moments were then provided as input to the \textsc{MomentInverse} algorithm, which returned an approximation to the spectrum.   The reason for repeating the spectrum approximation algorithm several times and and averaging the returned spectra was due to the tendency of the linear program to output sparsely supported spectra---perhaps due to the particular instabilities of Matlab's linear program solver.  Empirically, averaging several of these runs seemed to yield a very consistent spectrum that agreed closely with the ground truth for those networks on which we could compute the exact spectrum.

As the number of random walks was independent of the size of the graph, the runtime did not increase significantly for the larger graphs, and the computation time for each graph was at most 5 minutes and mostly is contributed to the optimization procedure which is independent of the graph.

For the smaller networks---those with $<50k$ nodes, we computed the exact spectrum in addition to running our spectrum estimation algorithm.  In all cases, our reconstruction achieved an earthmover distance at most $0.03$ from the actual spectrum.  For the larger networks, it was computationally intractable to compute the exact spectrum. 

\subsection{Discussion of Network Spectra}
The recovered spectra of the fifteen graphs considered are depicted in Figure~\ref{fig:exp}.   The emphasis of this work is the proposal of an efficient algorithm for recovering the spectrum, as opposed to a detailed analysis of the structural implications of the observed spectra of the graphs considered.  Nevertheless, the spectra exhibit several curious phenomena worth discussing.

The most immediate observations are that the spectra of the different classes of network look quite distinct, with the road networks exhibiting very distinctive linear spectra.   In hindsight, this should not be entirely unexpected.  Many portions of road networks resemble 2-d grids, and, for a random walk on a 2-d grid, the probability of returning to the origin after $t$ timesteps will scale roughly as $1/t$ for even $t$ (and will be 0 for odd $t$).  These return probabilities correspond to the moments of a uniform distribution supported on the interval $[-1,1]$, which is then translated to the uniform distribution over $[0,2]$ when the spectrum of the Laplacian is obtained from that of the random walk.

  The collaboration networks all have rather similar spectra, despite the DBLP network having a factor of 70 more nodes and edges than some of the other collaboration graphs.  This nicely illustrates the phenomena that certain classes of graph have spectra that approach a limiting shape, independent of their size.    

The spectra of the social networks appear more diverse.  One notable feature---particularly of the Google+, Orkut, and YouTube graphs is the significant number of eigenvalues that are extremely close to 1.  These eigenvalues correspond to eigenvectors near the kernel of the adjacency matrix, hence indicate that these adjacency matrices are significantly rank deficient.  In contrast to Facebook, Twitter, and LiveJournal where individuals tend to be more unique, perhaps many Google+ and YouTube users can be cleanly represented.

\begin{figure*}
\centering
\includegraphics[height=8.2in]{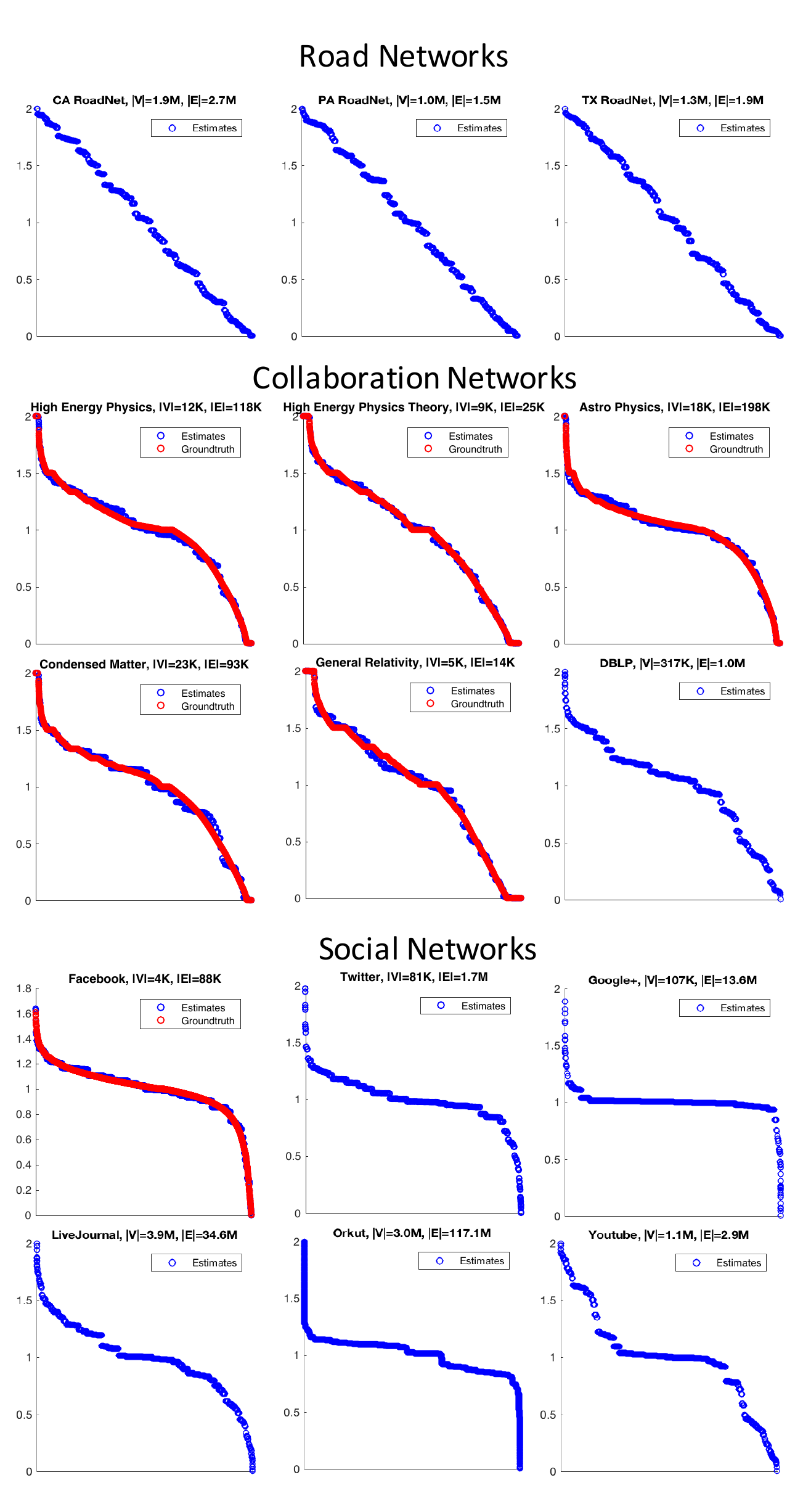}%[height=8.2in]
\caption{Plots of the results of running our spectrum estimation algorithm on 15 graphs that are publicly available from the Stanford Large Network Dataset Collection.%~\cite{snapnets}.  
For the graphs with $<50k$ nodes, the true spectrum (red) is superimposed on the estimated spectrum (blue).  All experiments were run in Matlab on a MacBook Pro laptop, and the estimated spectra required less than 5 minutes of computation time per graph.  Matlab code will be publicly available from our websites after the conclusion of the review process.\label{fig:exp}}
\end{figure*}

\newpage
\bibliographystyle{abbrv}
\bibliography{bibfile}

\newpage
\appendix

\end{document}